\documentclass[10pt]{article}
\usepackage[dvipsnames]{xcolor}
\usepackage{amssymb,color,tikz,amsthm,amsmath,amsfonts,xspace,caption}
\usetikzlibrary{decorations.markings}
\usetikzlibrary{arrows}
\usepackage{subcaption}
\usepackage{fullpage}
\usepackage[colorlinks,citecolor=blue,bookmarks=true]{hyperref}
\usepackage{aliascnt}
\usepackage[numbered]{bookmark}
\usepackage{verbatim,graphicx}
\usepackage{natbib}
\usepackage[boxed,titlenotnumbered]{algorithm2e}
\SetAlgoCaptionSeparator{} 

\usepackage{url}
\usepackage[textsize=scriptsize]{todonotes}
\usepackage{lipsum}
\usepackage{multirow}
\usepackage{longtable}
\usepackage{floatrow}
\usepackage{colortbl}
\usepackage{relsize}

\newcommand{\jnecess}[1]{\textcolor{black}{#1}}

\newcolumntype{L}[1]{>{\raggedright\let\newline\\\arraybackslash\hspace{0pt}}m{#1}}
\newcolumntype{C}[1]{>{\centering\let\newline\\\arraybackslash\hspace{0pt}}m{#1}}
\newcolumntype{R}[1]{>{\raggedleft\let\newline\\\arraybackslash\hspace{0pt}}m{#1}}

\newtheorem{theorem}{\textbf{Theorem}}
\newtheorem{lemma}[theorem]{\textbf{Lemma}}

\newtheorem{example}[theorem]{\textbf{Example}}

\newtheorem{definition}[theorem]{\textbf{Definition}}
\newtheorem{openq}[theorem]{\textbf{Open Question}}
\newcommand{\df}[1][\rm{def}]{\stackrel{#1}{=}}

\newtheorem{obs}{\textbf{Observation}}

\newcommand{\ignore}[1]{}

\newcommand{\iid}{\emph{i.i.d.}\xspace}
\newcommand{\ie}{\emph{i.e.,}\xspace}
\newcommand{\eg}{\emph{e.g.,}\xspace}

\def \sets#1{{\{#1\}}}

\newcommand{\cC}{\mathcal{C}}

\newcommand{\cA}{\mathcal{A}}
\newcommand{\cO}{\mathcal{O}}
\newcommand{\cY}{\mathcal{Y}}
\newcommand{\cX}{\mathcal{X}}
\newcommand{\cP}{\mathcal{P}}
\newcommand{\cS}{\mathcal{S}}
\newcommand{\cE}{\mathcal{E}}

\newcommand{\EE}{\mathbb{E}}

\newcommand{\pr}{\text{Pr}}

\newcommand{\norm}[1]{\left|\left|#1\right|\right|}
\newcommand{\argmin}{\operatornamewithlimits{argmin}}
\newcommand{\argmax}{\operatornamewithlimits{argmax}}

\newcommand{\algseq}{\algorithms{seq}}
\newcommand{\algcomp}{\algorithms{compl}}
\newcommand{\algcompsort}{\algorithms{compl-sort}}
\newcommand{\algkosub}{\algorithms{ko-sub}}
\newcommand{\algkomod}{\algorithms{ko-mod}}
\newcommand{\algqssub}{\algorithms{qs-sub}}
\newcommand{\algquickselect}{\algorithms{q-select}}
\newcommand{\algcomb}{\algorithms{comb}}
\newcommand{\algquicksort}{\algorithms{q-sort}}
\newcommand{\algscheffetest}{\textsc{scheffe-test}}

\ignore{
\newcommand{\algseq}{{\cA^{\textnormal{seq}}}}
\newcommand{\algcomp}{{\cA^{\textnormal{compl}}}}
\newcommand{\algscheffetest}{\textsc{Scheffe-test}}
\newcommand{\algcompsort}{{\cA^{\textnormal{compl-sort}}}}
\newcommand{\algkosub}{{\cA^{\textnormal{ko-sub}}}}
\newcommand{\algkomod}{{\cA^{\textnormal{ko-mod}}}}
\newcommand{\algqssub}{{\cA^{\textnormal{qs-sub}}}}
\newcommand{\algquickselect}{{\cA^{\textnormal{q-select}}}}
\newcommand{\algcomb}{{\cA^{\textnormal{comb}}}}
\newcommand{\algquicksort}{{\cA^{\textnormal{q-sort}}}}
}

\newcommand{\maxx}{x^*}
\newcommand{\bY}{\textbf{Y}}
\newcommand{\din}{\textnormal{in}}
\newcommand{\dout}{\textnormal{out}}
\newcommand{\nonadaptivecomp}{\cC_{\textnormal{non}}}
\newcommand{\adaptivecomp}{\cC_{\textnormal{adpt}}}
\newcommand{\pe}{\cE}
\newcommand{\adaptive}{\textnormal{adpt}}
\newcommand{\nonadaptive}{\textnormal{non}}
\newcommand{\thrr}{\Delta}
\newcommand{\deltacover}{\delta}
\newcommand{\errorcomb}{\epsilon}
\newcommand{\errordens}{\varepsilon}
\newcommand{\errorquicksort}{\epsilon}

\newcommand\blfootnote[1]{%
  \begingroup
  \renewcommand\thefootnote{}\footnote{#1}%
  \addtocounter{footnote}{-1}%
  \endgroup
}

\title{
Maximum Selection and Sorting with Adversarial Comparators and an Application to Density Estimation
\blfootnote{Part of this paper appeared in \cite{AcharyaJOS14e}.}}
\author{Jayadev Acharya\\EECS, MIT\\\tt{jayadev@csail.mit.edu}
\and
Moein Falahatgar\\ECE, UCSD\\\tt{moein@ucsd.edu}
\and
Ashkan Jafarpour\\Yahoo Labs\\\tt{ashkanj@yahoo-inc.com}
\and
Alon Orlitksy\\ECE \& CSE, UCSD\\\tt{alon@ucsd.edu}
\and
Ananda Theertha Suresh\\Google Research\\\tt{s.theertha@gmail.com}}

\begin{document}
\maketitle

\begin{abstract}
We study maximum selection and sorting of $n$ numbers using pairwise
comparators that output the larger of their two inputs if the inputs
are more than a given threshold apart, and output an
adversarially-chosen input otherwise.
We consider two adversarial models. A non-adaptive adversary
that decides on the outcomes in advance based solely on the inputs, 
and an adaptive adversary that can decide on the outcome of each query
depending on previous queries and outcomes. 

Against the non-adaptive adversary, we derive a maximum-selection
algorithm that uses at most $2n$ comparisons in expectation, and a
sorting algorithm that uses at most $2n\ln n$ comparisons in
expectation. These numbers are within small constant factors from the
best possible.  Against the adaptive adversary, we propose a
maximum-selection algorithm that uses $\Theta(n\log (1/{\errorcomb}))$
comparisons to output a correct answer with probability at least
$1-\errorcomb$.  The existence of this algorithm affirmatively
resolves
an open problem of Ajtai, Feldman, Hassadim, and
Nelson~\cite{AjtaiFHN15}.

Our study was motivated by a density-estimation problem where, given
samples from an unknown underlying distribution, we would like to find
a distribution in a known class of $n$ candidate distributions that is
close to underlying distribution in $\ell_1$ distance.
Scheffe's algorithm~\cite{DevroyeL01} outputs a distribution at an
$\ell_1$ distance at most 9 times the minimum and runs in time
$\Theta(n^2\log n)$.  Using maximum selection, we propose an algorithm
with the same approximation guarantee but run time of $\Theta(n\log n)$.
\end{abstract}

\section{Introduction}
\subsection{General background and motivation}
Maximum selection and sorting are fundamental operations with wide-spread
applications in computing, investment, marketing~\cite{AMPP09}, decision
making~\cite{Thurstone27,David63}, and sports.
These operations are often accomplished via pairwise
comparisons between elements, and the goal is to minimize the
number of comparisons.

For example, one may find the largest of $n$ elements by first
comparing two elements and then successively comparing the larger one
to a new element.  This simple algorithm takes $n-1$ comparisons, and
it is easy to see that $n-1$ comparisons is necessary.  Similarly,
\emph{merge sort} sorts $n$ elements using less than $n\log n$
comparisons, close to the information theoretic lower bound of $\log
n!=n\log n -o(n)$.

However, in many applications, the pairwise comparisons may be
imprecise. For example, when comparing two random numbers, such
as stock performances, or team strengths, the output of the
comparison may vary due to chance. 
Consequently, a number of researchers have considered
maximum selection and sorting with imperfect, or noisy, comparators.

The comparators in these models mostly function correctly, but
occasionally may produce an inaccurate comparison result,
where the form of inaccuracy is dictated by the application.
The Bradley-Terry-Luce~\cite{BradleyT52} model assumes that if two
values $x$ and $y$ are compared, then $x$ is selected as the larger
with probability $x/(x+y)$.
Observe that the comparison is correct with probability $\max\{x,y\}/(x+y)\ge 1/2$. 
Algorithms for ranking and estimating weights under this model were
proposed, \eg in~\cite{NegahbanOS12}.
Another model  assumes that the output of any comparator gets reversed
with probability less than $1/2$.
Algorithms applying this model for maximum selection were proposed
in~\cite{AdlerGHKK94} and for ranking in~\cite{KarpK07,BravermanM08}.

\ignore{Other related works in the noisy comparator setting include the
\emph{feedback arc set  problem},~\cite{Alon06,AilonCN08,CoppersmithFR10},
subset ranking~\cite{CossockZ08}, top-$k$ ranking~\cite{XiaLL09}, and
consistency of general ranking under weighted graph
queries~\cite{DuchiMJ10}.}

We consider a third  model where, unlike the previous models, the
comparison outcome can be adversarial.
If the numbers compared are more than a threshold $\thrr$ apart,
the comparison is correct, while if they differ by at most $\thrr$,
the comparison is arbitrary, and possibly even adversarial. 

This model can be partially motivated by physical observations.
Measurements are regularly quantized and often adulterated with
some measurement noise. Quantities with the same quantized value may
therefore be incorrectly compared. In psychophysics, the
Weber-Fechner law~\cite{Ekman59} stipulates that humans can
distinguish between two physical stimuli only when their difference
exceeds some threshold (known as \emph{just
    noticeable difference}).  And in sports, a judge or a home-team
advantage may, even adversarially, sway the outcome of a game between
two teams of similar strength, but not between teams of significantly
different strengths.  Our main motivation for the model
derives from the important problem of density-estimation and
distribution-learning.

\subsection{Density estimation via pairwise comparisons}
In a typical PAC-learning setup~\cite{Valiant84,KearnsMRRSS94},
we are given samples from an unknown distribution $p_0$ in a
known distribution class $\cP$ and would like to find,
with high probability, a distribution $\hat p\in\cP$ such that
$\|\hat p-p_0\|_1<\deltacover$.

One standard approach proceeds in two steps~\cite{DevroyeL01}.
\begin{enumerate}
\item 
Offline, construct a \emph{$\deltacover$-cover} of $\cP$, a finite
collection $\cP_{\deltacover}\subseteq\cP$ of distributions such that
for any distribution $p\in\cP$, there is a distribution
$q\in\cP_{\deltacover}$ such that $\|p-q\|_1<\deltacover$.
\item
Using the samples from $p_0$, find a distribution in
$\cP_{\deltacover}$ whose $\ell_1$ distance to $p_0$ is close to the
$\ell_1$ distance of the distribution in $\cP_{\deltacover}$ that is
closest to $p_0$.
\end{enumerate}
These two steps output a distribution whose $\ell_1$
distance from $p_0$ is \emph{close} to $\deltacover$.
Surprisingly, for several common distribution classes, such as Gaussian
mixtures, the number of samples required by this generic approach
matches the information theoretically optimal sample complexity,
up to logarithmic factors~\cite{DaskalakisK13, AcharyaJOS14, diako16}.

The Scheffe Algorithm~\cite{Scheffe47, DevroyeL01} is a popular method
for implementing the second step, namely to find a distribution in
$\cP_{\deltacover}$ with a small $\ell_1$ distance from $p_0$.  It
takes every pair of distributions in $\cP_\deltacover$ and uses the
samples from $p_0$ to decide which of the two distributions is closer
to $p_0$.  It then declares the distribution that ``wins'' the most
pairwise closeness comparisons to be the nearly-closest
to $p_0$.  As shown in~\cite{DevroyeL01}, with high probability, the
Scheffe algorithm yields a distribution that is at most 9 times
further from $p_0$ than the distribution in $\cP_\deltacover$
  with the lowest $\ell_1$ distance from $p_0$, plus a diminishing
  additive term; hence finds a distribution that is roughly $9\delta$
  away from $p_0$. Since this algorithm compares every pair of
  distributions in $\cP_\deltacover$, it uses quadratic in
  $|\cP_\deltacover|$ comparisons.  In Section~\ref{sec:densityestimation}, we
use maximum-selection results to derive an algorithm with the same
approximation guarantee but with linear in
$|\cP_\deltacover|$ comparisons.

\subsection{Organization}
The paper is organized as follows.  In
Section~\ref{sec:problemstatement} we define the problem and introduce
the notations, in Section~\ref{sec:results} we summarize the results,
  in Section~\ref{sec:simple-results} we derive simple bounds and
  describe the performance of simple algorithms, and in
  Section~\ref{sec:algorithm} we present our main maximum-selection
  algorithms.  The relation between density estimation problem and our
  comparison model is discussed in
  Section~\ref{sec:densityestimation}, and in
  Section~\ref{sec:sorting} we discuss sorting with adversarial
  comparators.

\section{Notations and Preliminaries}
\label{sec:problemstatement}
Practical applications call for sorting or selecting the maximum of
not just numbers, but rather of items with associated
values. For example, finding the person with the highest
salary, the product with lowest price, or a sports team with
  the most \emph{capability} of winning. 
Associate with each item $i$ a real value
  $x_i$, and let $\cX\df\{x_1,\ldots,x_n\}$ be the multiset of
values.  In maximum selection, we use noisy pairwise comparisons to
find an index $i$ such that $x_i$ is close to the largest element
$\maxx\df\max\{x_1,\ldots,x_n\}$.

Formally, a faulty comparator $\cC$ takes two distinct indices $i$ and $j$, and
if $|x_i-x_j|>\thrr$, outputs the index associated with the higher value,
while if $|x_i-x_j|\le\thrr$, outputs either $i$ or $j$, possibly
adversarially. 
Without loss of generality, we assume that $\thrr=1$. Then, 
\[
\cC(i,j)= \left\{
\begin{array}{ll}
\arg\max\sets{x_i,x_j} & \text{if\qquad} |x_i-x_j|>1,\\
i \text{ or } j \text{ (adversarially)} & \text{if\qquad} |x_i-x_j|\le1.
\end{array} \right.
\]
It is easier to think just of the numbers, rather than the indices.
Therefore, informally we will simply view 
the comparators as taking two real inputs $x_i$ and $x_j$, and outputting
\begin{equation}
\label{eq:comp}
\cC(x_i,x_j)= \left\{
\begin{array}{ll}
\max \{x_i,x_j\} & \text{if\qquad} |x_i-x_j|>1,\\
x_i \text{ or } x_j \text{ (adversarially)} & \text{if\qquad} |x_i-x_j| \le 1.
\end{array} \right.
\end{equation}

We consider two types of adversarial comparators: \emph{non-adaptive}
and \emph{adaptive}.
\begin{itemize}
\item 
A \emph{non-adaptive adversarial comparator} that has complete knowledge of $\cX$ and the algorithm, but must fix its outputs for every pair of inputs before the algorithm starts. 
\item 
An \emph{adaptive adversarial comparator} not only has access to the algorithm and the inputs, but is also allowed to adaptively decide the outcomes of the queries taking into account all the previous queries made by the algorithm.
\end{itemize}

  A non-adaptive comparator can be naturally represented by a
  directed graph with $n$ nodes representing the $n$ indices. There is
  an edge from node $i$ to node $j$ if the comparator declares $x_i$
  to be larger than $x_j$, namely, $\cC(x_i,x_j)=x_i$.
  Figure~\ref{fig:comparison_4} is an example of such a
  comparator, where for simplicity we show only the values 0, 1, 1, 2,
  and not the indices.  Note that by definition, $\cC(2,0)=2$, but for
  all the other pairs, the outputs can be decided by the comparator. In
this example, the comparator declares the node with value 2 as the
``winner'' against the right node with value 1, but as the ``loser''
against the left node, also with value 1. Among the two nodes
  with value 1, it arbitrarily declares the left one as the winner.
An adaptive adversary reveals the edges one by one as the algorithm
proceeds.
\begin{figure}[h]
\begin{center}
\begin{tikzpicture}[scale=1.0]
 \footnotesize
 \node (x0) at (90:1){$2$};
 \node (x1) at (0:1){$1$};
 \node (x2) at (270:1){$0$};
 \node (x3) at (180:1) {$1$};
 \draw [arrows={-angle 60}](x0) to (x2);
 \draw [arrows={-angle 60}](x0) to (x1);
 \draw [arrows={-angle 60}](x3) to (x0);
 \draw [arrows={-angle 60}](x3) to (x1);
 \draw [arrows={-angle 60}](x3) to (x2);
 \draw [arrows={-angle 60}](x2) to (x1);
 \normalsize
\end{tikzpicture}
\end{center}
\caption{Comparator for four inputs with values $\sets{0,1,1,2}$}
\label{fig:comparison_4}
\end{figure}
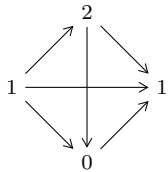

We refer to each comparison as a \emph{query}.  The number of queries
an algorithm $\cA$ makes for $\cX=\{x_1,\ldots, x_n\}$ is its
\emph{query complexity}, denoted by $Q_n^{\cA}$.\footnote{This
      is a slight abuse of notation suppressing $\cX$.} Our algorithms
  are randomized and $Q_n^{\cA}$ is a random variable. The
  \emph{expected query complexity} of $\cA$ for the input $\cX$ is
\[
q_n^{\cA} \df \EE[Q_n^{\cA}],
\]
where the expectation is over the randomness of the algorithm.

Let $\nonadaptivecomp(\cX)$ or simply $\nonadaptivecomp$  be the
set of all
non-adaptive adversarial comparators, and let $\adaptivecomp$ be
the set of all adaptive adversarial comparators.
The \emph{maximum expected query complexity} of $\cA$
against non-adaptive adversarial comparators is  
\begin{align}
\label{eqn:nonadaptivemodel}
 q_n^{\cA,\nonadaptive}\df\max_{\cC\in\nonadaptivecomp}\max_{\cX} q_n^{\cA}.
\end{align}
Similarly, the maximum expected
query complexity of $\cA$ against adaptive adversarial comparators is
\begin{align*}
 q_n^{\cA,\adaptive}\df\max_{\cC\in\adaptivecomp}\max_{\cX} q_n^{\cA}.
\end{align*}

We evaluate an algorithm by how close its output is to $\maxx$, the
maximum of $\cX$.  
\begin{definition} A number $x$ is a
    \emph{$t$-approximation} of $\maxx$ if $x\ge\maxx-t$.
\end{definition}

\noindent The \emph{$t$-approximation error} of an algorithm $\cA$ over $n$ inputs is
\[
\pe_n^\cA(t)\df \pr\left(Y_\cA(\cX) <\maxx- t\right),
\]
the probability that $\cA$'s output $Y_\cA(\cX)$ is \emph{not} a 
$t$-approximation of $\maxx$.
For an algorithm~$\cA$, the maximum $t$-approximation error for the
worst non-adaptive adversary  is
\[
 \pe_n^{\cA,\nonadaptive}(t)\df \max_{\cC\in\nonadaptivecomp} \max_{\cX} \pe_n^{\cA}(t),
\]
and similarly for the adaptive adversary,
\[
 \pe_n^{\cA,\adaptive}(t)\df  \max_{\cC\in\adaptivecomp} \max_{\cX} \pe_n^{\cA}(t).
\]
For the non-adaptive adversary, the minimum $t$-approximation error of any algorithm is 
\[
 \pe_n^{\nonadaptive}(t)\df \min_{\cA} \pe_n^{\cA,\nonadaptive}(t).
\]
and similarly for the adaptive adversary,
\[
 \pe_n^{\adaptive}(t)\df \min_{\cA} \pe_n^{\cA,\adaptive}(t).
\]
Since adaptive adversarial comparators are stronger than non-adaptive, for all $t$,
$$\pe_n^{\adaptive}(t)\ge\pe_n^{\nonadaptive}(t).$$

The next example shows that $\pe_3^{\nonadaptive}(t)\ge\frac13$ for all $t<2$.
\begin{example}
\label{exm:one}
$\pe_3^{\nonadaptive}(t)\ge\frac13$ for all $t<2$.
Consider $\cX=\{0,1,2\}$ and the following comparators. 
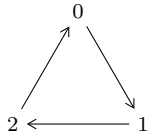
\begin{center}
\begin{tikzpicture}[scale=1.0]
 \footnotesize
 \node (x0) at (90:1){$0$};
 \node (x1) at (-30:1) {$1$};
 \node (x2) at (210:1){$2$};
 \draw [arrows={-angle 60}](x0) to (x1);
 \draw [arrows={-angle 60}](x1) to (x2);
 \draw [arrows={-angle 60}](x2) to (x0);
 \normalsize
\end{tikzpicture}
\end{center}

\ignore{Since extra information can only help the algorithm,
consider an algorithm that has the result of all the $3$
pairwise comparisons. This algorithm knows that there exists three
numbers $x_1,x_2,x_3$ such that each one is bigger than one of the
remaining two numbers and smaller than the other one. }
By symmetry, no algorithm can differentiate between the three inputs,
hence any algorithm will output 0 with probability $1/3$.

\end{example}

\section{Previous and new results}
\label{sec:results}
In Section~\ref{subsec:lowerbound} we lower bound $\pe_n^{\nonadaptive}(t)$ as a function of $t$.
In Lemma~\ref{lem:oneerror}, we show that for all $t<1$ and odd $n$, $\pe_n^{\nonadaptive}(t)\ge 1-1/n$, namely for some $\cX$,
approximating the maximum to within less than one is equivalent to 
guessing a random $x_i$ as the maximum.
In Lemma~\ref{lem:twoerror}, we modify 
Example~\ref{exm:one} and show that for all $t<2$ and odd $n$, any 
algorithm has $t$-approximation error close
to $1/2$ for some input.  

We propose a number of algorithms to approximate the
maximum. These algorithms have different  guarantees in
terms of the probability of error, approximation factor, and query
complexity.  

We first consider two simple algorithms, the complete tournament,
denoted $\algcomp$, and the sequential selection,
denoted $\algseq$. Algorithm $\algcomp$ compares
all the possible input pairs, and declares the input with the
most number of wins as the maximum.  We show the simple result
that almost surely $\algcomp$ outputs a $2$-approximation of
$\maxx$.  We then consider the algorithm $\algseq$ that
compares a pair of inputs and discards the loser and compares the
winner with a new input. We show that even under random selection of
the inputs, there exist inputs such that, with high probability,
$\algseq$ cannot provide a constant approximation to $\maxx$.

  We then consider more advanced algorithms.
  The knock-out algorithm, at each stage pairs the
  inputs at random, and keeps the winners of the comparisons for the
  next stage. We design a slight modification of this algorithm,
  denoted $\algkomod$ that achieves a $3$-approximation with error
  probability at most $\errorcomb$, even against adaptive adversarial
  comparators.  We note that~\cite{AjtaiFHN15} proposed a different
algorithm with similar performance guarantees.

Motivated by quick-sort, we propose a quick-select algorithm
$\algquickselect$ that outputs a $2$-approximation, with zero
  error probability. It has an expected query complexity of at most
  $2n$ against the non-adaptive adversary. However, in
  Example~\ref{exm:expected_quick_select}, we see that this algorithm
  requires $\binom n2$ queries against the adaptive adversary.

This leaves the question of whether there is a randomized algorithm
for $2$-approximation of $\maxx$ with $\cO(n)$ queries against the
adaptive adversary. In fact,~\cite{AjtaiFHN15} pose this as an open
question. We resolve this problem by designing an algorithm $\algcomb$
that combines quick-select and knock-out. We prove that $\algcomb$
outputs a $2$-approximation with probability of error at most
$\errorcomb$, using $\cO(n\log\frac{1}{\errorcomb})$ queries. We
summarize the results in Table~\ref{tab:results}.
\begin{table}[ht]
 \centering
 \begin{tabular}{| L{4.0cm}| c | c | c| c |}
    \hline
    algorithm& notation& approximation & $q_n^{\cA,\nonadaptive}$ & $q_n^{\cA,\adaptive}$  \\ \hline
    complete tournament & $\algcomp$& $\pe_n^{\algcomp,\adaptive}(2)=0$ & \multicolumn{2}{c|}{$\binom n2$}  \\ \hline
    deterministic upper bound~\cite{AjtaiFHN15}& - & $\pe_n^{\cA,\adaptive}(2)=0$& \multicolumn{2}{c|}{$\Theta(n^{\frac32})$}  \\ \hline
    deterministic lower bound~\cite{AjtaiFHN15}& -& $\pe_n^{\cA,\adaptive}(2)=0$  & - & $\Omega(n^{\frac43})$ \\ \hline
    sequential& $\algseq$ & $\pe_n^{\algseq,\nonadaptive}\left(\tfrac{\log n}{\log \log n}-1\right)\to 1$& \multicolumn{2}{c|}{$n-1$} \\ \hline
    modified knock-out  & $\algkomod$& $\pe_n^{\algkomod,\adaptive}(3) < \epsilon$ &\multicolumn{2}{c|}{$< n+\frac12 \log^4 n \left\lceil\frac1\errorcomb\ln \frac1\errorcomb \right\rceil^2$} \\ \hline
    quick-select & $\algquickselect$ & $\pe_n^{\algquickselect,\adaptive}(2)=0$ & $<2n$ & $\binom n2$ \\ \hline
    knock-out and quick-select combination & $\algcomb$ & $\pe_n^{\algcomb,\adaptive}(2)<\errorcomb$  & \multicolumn{2}{c|}{$\cO\left(n\log \tfrac1\errorcomb\right)$}\\ \hline
  \end{tabular}
  \caption{Maximum selection algorithms}
  \label{tab:results}
\end{table}

We note that while we focus on randomized
algorithms,~\cite{AjtaiFHN15} also studied the best possible
trade-offs for deterministic algorithms.  They designed a
deterministic algorithm for $2$-approximation of the maximum using
only $\cO(n^{3/2})$ queries. Moreover, they prove that no
deterministic algorithm with fewer than $\Omega(n^{4/3})$ queries can
output a $2$-approximation of $\maxx$ for the adaptive adversarial
model.

\section{Simple results}
\label{sec:simple-results}
In Lemmas~\ref{lem:oneerror} and~\ref{lem:twoerror} we prove lower
bounds on the error probability of any algorithm that provides a
$t$-approximation of $\maxx$ for $t<1$ and $t<2$ respectively.  We
then consider two straightforward algorithms
for finding the maximum. One is the complete tournament, where each
pair of inputs is compared, and the other is sequential, where inputs are compared sequentially and the loser is
discarded at each comparison.

\subsection{Lower bounds}
\label{subsec:lowerbound}
We show the following two lower bounds.
\begin{itemize}
\item $\pe_n^{\nonadaptive}(t)\ge 1-\frac1{n}$ for all $0\le t<1$ and odd $n$. 
\item $\pe_n^{\nonadaptive}(t)\ge \frac12-\frac1{2n}$ for all $1\le t<2$ and odd $n$.
\end{itemize}
These lower bounds can be applied to $n$ which is even, by adding an
extra input smaller than all the others and losing to everyone.
\begin{lemma}
 \label{lem:oneerror}
For all $0 \le t<1$ and odd $n$,
 $$\pe_n^{\nonadaptive}(t)\ge 1-\frac1n.$$
\end{lemma}
\begin{proof}
 Let $(x_1,x_2,\ldots,x_n)$ be an unknown permutation of
 $(1,\underbrace{0,\ldots,0}_{n-1})$. Suppose we consider an
 adversary that ensures each input wins exactly $(n-1)/2$ times. An
 example is shown in Figure~\ref{fig:tournament_for_error_one} for
 $n=5$. 
\begin{figure}[h]
\begin{center}
\begin{tikzpicture}[scale=1.5]
 \footnotesize
 \node (x1) at (162:1){$1$};
 \node (x2) at (90:1){$0$};
 \node (x3) at (18:1){$0$};
 \node (x4) at (-54:1){$0$};
 \node (x5) at (-126:1){$0$};
 \draw [arrows={-angle 60}] (x1) to (x2);
 \draw [arrows={-angle 60}](x1) to (x3);
 \draw [arrows={-angle 60}](x2) to (x3);
 \draw [arrows={-angle 60}](x2) to (x4);
 \draw [arrows={-angle 60}](x3) to (x4);
 \draw [arrows={-angle 60}](x3) to (x5);
 \draw [arrows={-angle 60}](x4) to (x5);
 \draw [arrows={-angle 60}](x4) to (x1);
 \draw [arrows={-angle 60}](x5) to (x1);
 \draw [arrows={-angle 60}](x5) to (x2);
 \normalsize
\end{tikzpicture}
\end{center}
 \caption{Tournament for Lemma~\ref{lem:oneerror} when $n=5$}
 \label{fig:tournament_for_error_one}
\end{figure}
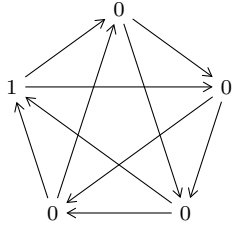

We want a lower bound on the performance of any randomized
algorithm. By Yao's principle, we consider only deterministic
algorithms over a uniformly chosen permutation of the inputs, namely
only one of the coordinates is 1, and remaining are less than $1-t$. In this case,
if we fix any comparison graph (as in the figure above), and permute
the inputs, the algorithm cannot distinguish between 1 and $0$'s,
and outputs $0$ with probability $1-1/n$.
\end{proof}
\begin{lemma}
 \label{lem:twoerror}
For all $1 \le t<2$ and odd $n$,
$$\pe_n^{\nonadaptive}(t)\ge \frac12-\frac1{2n}.$$ 
\end{lemma}
\begin{proof}
Let $m$ be $(n-1)/2$. Let $(x_1,x_2,\ldots,x_n)$ be an unknown permutation of
$(2,\underbrace{1,\ldots,1}_{m},\underbrace{0,\ldots,0}_{m})$.
Suppose the adversary ensures that 2 loses
against all the $1$'s and indeed all inputs have exactly $(n-1)/2$
wins. An example is shown in Figure~\ref{fig:tournament_for_error_two}. 
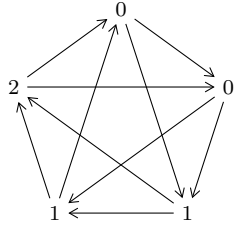
\begin{figure}[h]
\begin{center}
\begin{tikzpicture}[scale=1.5]
 \footnotesize
 \node (x1) at (162:1){$2$};
 \node (x2) at (90:1){$0$};
 \node (x3) at (18:1){$0$};
 \node (x4) at (-54:1){$1$};
 \node (x5) at (-126:1){$1$};
 \draw [arrows={-angle 60}](x1) to (x2);
 \draw [arrows={-angle 60}](x1) to (x3);
 \draw [arrows={-angle 60}](x2) to (x3);
 \draw [arrows={-angle 60}](x2) to (x4);
 \draw [arrows={-angle 60}](x3) to (x4);
 \draw [arrows={-angle 60}](x3) to (x5);
 \draw [arrows={-angle 60}](x4) to (x5);
 \draw [arrows={-angle 60}](x4) to (x1);
 \draw [arrows={-angle 60}](x5) to (x1);
 \draw [arrows={-angle 60}](x5) to (x2);
 \normalsize
\end{tikzpicture}
\end{center}
 \caption{Tournament for Lemma~\ref{lem:twoerror} when $n=5$}
 \label{fig:tournament_for_error_two}
\end{figure}

Similar to Lemma~\ref{lem:oneerror}, the inputs are all identical to the algorithm and therefore, 
the algorithm outputs one of the $0$'s with probability $\frac{m}{n}=\frac12-\frac1{2n}$. 
\end{proof}
\subsection{Two elementary algorithms}
\label{sec:elementary}
In this section, we analyze two familiar maximum selection algorithms,
the complete tournament and sequential selection. We discuss their
strengths and weaknesses, and show that there is a trade-off between
the query complexity and the approximation guarantees
of these two algorithms. Another well-known algorithm for maximum
selection is knock-out algorithm and we discuss a variant of
it in Section~\ref{subsec:knockout}.
\subsubsection{Complete tournament (round-robin)}
\label{subsec:complete}
As its name evinces, a complete tournament involves a match between
every pair of teams. Using this metaphor to competitions, we
compare all the  $n \choose 2$ input pairs, and the input winning
the maximum number of times is declared as the output. If two or more inputs end up with the highest number of 
wins, any of them can be declared as the output. This algorithm is
formally stated in $\algcomp$. 

\begin{algorithm}
\caption{$\algcomp$ - Complete tournament}
\textbf{input:} $\cX$\\
\qquad compare all input pairs in $\cX$, count the number of times each input wins\\
\textbf{output:} an input with the maximum number of wins 
\end{algorithm}

The next lemma shows that this algorithm gives a 2-approximation
against both adversaries. The result, although weaker than the
deterministic guarantees of~\cite{AjtaiFHN09}, is illustrative and
useful in the algorithms proposed later. 
\begin{lemma}
 \label{lem:complete-tournament}
$q_n^{\algcomp,\adaptive}=\binom{n}{2}$ and $\pe_n^{\algcomp,\adaptive}(2)=0$. 
\end{lemma}
\begin{proof}
The number of queries is clearly ${n\choose 2}$.  To show that
$\pe_n^{\algcomp,\adaptive}(2)=0$, note that if $y<\maxx-2$ then
for all $z$ that $y$ wins over, $z\le y+1 <\maxx-1$, and
  therefore $\maxx$ also beats them. Since $\maxx$ wins over $y$, it
  wins over more inputs than $y$ and hence $y$ cannot be the output of
  the algorithm. It follows that the input with the maximum
number of wins is a $2$-approximation of $\maxx$.
\end{proof}

\jnecess{$\algcomp$ is deterministic and after $\binom n2$ queries
  it outputs a $2$-approximation of $\maxx$.  If the
  comparators are noiseless, we can simply compare the inputs
  sequentially, discarding the loser at each step, thus requiring only
  $n-1$ comparisons.  This evokes the hope of finding a deterministic
  algorithm that requires a linear number of comparisons and outputs a
  $2$-approximation of $\maxx$. As mentioned earlier,
  \cite{AjtaiFHN15} showed it is not achievable.  They proved that any
  deterministic $2$-approximation algorithm requires $\Omega(n^{4/3})$
  queries. They also showed a strictly superlinear lower bound on any
  deterministic constant-approximation algorithm. On the other hand,
  they designed a deterministic $2$-approximation algorithm using
  $\cO(n^{3/2})$ queries.}

\subsubsection{Sequential selection}
\label{subsec:sequential}
Sequential selection first compares a random pair of inputs, and at each
successive step, compares the winner of the last comparison with
a randomly chosen \emph{new} input. It outputs the final remaining input. 
This algorithm uses $n-1$ queries.

\begin{center}
\begin{algorithm}
\caption{$\algseq$ - Sequential selection}
\textbf{input:} $\cX$\\
\qquad choose a random $y\in \cX$ and remove it from $\cX$\\
\qquad\textbf{while} $\cX$ is not empty\\
\qquad\qquad choose a random $x\in \cX$ and remove it from $\cX$\\
\qquad\qquad 	$y\leftarrow\cC(x,y)$\\
\qquad\textbf{end while}\\
\textbf{output:} $y$
\end{algorithm}
\end{center}

The next lemma shows that even against the non-adaptive adversarial
comparators, the algorithm cannot output a constant-approximation of
$\maxx$.

\begin{lemma}
\label{lem:sequential_example}
 Let $s=\frac{\log n}{\log\log n}$. For all $t<s$,
 $$\pe_n^{\algseq,\nonadaptive}(t)\ge 1-\frac{1}{\log \log n}.$$
\end{lemma}
\begin{proof}
Assume that $s$, $\log n$, and $\log\log n$ are integers and
\[
x_i=\left\{
 \begin{array}{ll}
  s & \text{for }i=1,\\
  s-1 &\text{for } i=2,\ldots,r,\\
  s-2 &\text{for } i=r+1,\ldots,r^2,\\
  \vdots\\
  m & \text{for } i=r^{s-m-1}+1,\ldots,r^{s-m},\\
  \vdots\\
  0 & \text{for } i=r^{s-1}+1,\ldots,r^s,
 \end{array}\right.
\]
where $r=\log n$.
Consider the following non-adaptive adversarial comparator,
\begin{align}\label{eqn:adv}
\cC(x_i,x_j) = \left\{
\begin{array}{ll}
\max\{x_i,x_j\} & \text{if } |x_i-x_j|>1,\\
\min\{x_i,x_j\} & \text{if }  |x_i-x_j|\le 1.
\end{array} \right.
\end{align}

The sequential algorithm takes a random permutation of the inputs. It
then starts by comparing the first two elements, and then sequentially
compares the winner with the next element, and so on. Let $L_j$ be the
location in the permutation where input $j$ appears for the last time.
The next two observations follow from the construction of inputs and
comparators respectively.
\begin{obs}
\label{obs:thenumber}
Input $j$ appears at least $(\log n-1)$ times more than input $j+1$.
\end{obs}
\begin{obs}
 \label{obs:forthe}
 For the adversarial comparator defined in~\eqref{eqn:adv}, if
 $L_0>L_1>\ldots>L_s$ then $\algseq$ outputs $0$.
\end{obs}

As a consequence of Observation~\ref{obs:thenumber}, in the random permutation of inputs, $L_j>L_{j+1}$ with probability at least $1-\frac1{\log n}$. By the union bound, $L_0>L_1>\ldots>L_s$ with probability at least,
\[
1-\frac{s}{\log n}=1-\frac1{\log \log n}.
\]
By applying Observation~\ref{obs:forthe}, $\algseq$ outputs $0$ with probability at least $1-\frac1{\log \log n}$.
\end{proof}

\section{Algorithms}
\label{sec:algorithm}
In the previous section we saw that the complete tournament,
$\algcomp$, always outputs a $2$-approximation, but has quadratic query
complexity, while the sequential selection, $\algseq$, has linear
query complexity but a poor approximation guarantee. 
A natural question to ask is whether the benefits of these two
algorithms
can be combined to derive bounded error with linear query complexity.
In this section, we propose algorithms with linear query
complexity
and approximation guarantees that compete with the best possible, \ie $2$-approximation of $\maxx$.

We propose three algorithms, with varying
performance guarantees:
\begin{itemize}
\item
{\bf Modified knock-out}, described in Section~\ref{subsec:knockout},
has linear query complexity and with high probability outputs a
$3$-approximation of $\maxx$ against both  the adaptive and non-adaptive adversaries.
\item
{\bf Quick-select}, described in
Section~\ref{subsec:quickselect}, outputs a $2$-approximation to $\maxx$ (against both adversaries). It also has a linear expected query complexity against  non-adaptive adversarial comparators. 
\item
{\bf Knock-out and quick-select combination}, described 
in Section~\ref{sec:combined}, has linear query
complexity, and with high probability outputs a $2$-approximation of
$\maxx$ even against  adaptive adversarial comparators. 
\end{itemize}

\noindent We now go over these algorithms in detail.
\subsection{Modified knock-out}
\label{subsec:knockout}
For simplification, in this section we assume that $\log n$ is
an integer.  The knock-out algorithm derives its name from knock-out
competitions where the tournament is divided into $\log n$ successive
rounds. In each round the inputs are paired at random and the winners
advance to the next round. Therefore, in round $i$ there are
$\frac{n}{2^{i-1}}$ inputs. The winner at the end of $\log n$
rounds is declared as the maximum.

Under our adversarial model, at each round of the knock-out
algorithm, the largest remaining input decreases by at most one.
Therefore, knock-out algorithm finds at least $\log n$-approximation
of $\maxx$.  Analyzing the precise approximation error of knock-out
algorithm appears to be difficult.  \jnecess{However,
  simulations suggest that for any large $n$, for the set consisting
  of $0.2\cdot n$ 0's, $\alpha\cdot n$ 1's, $(0.7-\alpha)\cdot n$ 2's,
  $0.1\cdot n$ 3's, and a single 4, where $0<\alpha<0.7$ is an
  appropriately chosen parameter, the knock-out algorithm is not able
  to find a $3$-approximation of $\maxx$ with positive constant
  probability. The problem with knock-out algorithm is that
  if at any of the $\log n$ rounds, many inputs are within $1$ from
  the largest input at that round, there is a fair chance that the
  largest input will be eliminated. If this elimination happens in
  several rounds, we will end up with a number significantly
  smaller than $\maxx$.}

To circumvent the problem of discarding large inputs, we select a specified number of inputs at each round and save them for the very end, thereby ensuring that at every round, if the largest input is eliminated, then an input within 1 from it has been saved.
We then perform a complete tournament on these saved inputs.
The algorithm is explained in $\algkomod$.

\begin{algorithm}
\caption{$\algkosub$ - Subroutine for $\algkomod$ and $\algcomb$}
\textbf{input:} $\cX$\\
\qquad pair the inputs of $\cX$ randomly, let $\cX'$ be the winners\\
\textbf{output:} $\cX'$ 
\end{algorithm}

\begin{algorithm}
\caption{$\algkomod$ - Modified knock-out algorithm}
\textbf{input:} $\cX,\errorcomb$\\
\qquad $\cY=\emptyset$, $n_1 =\left\lceil\frac1\errorcomb\ln \frac1\errorcomb \cdot\log n \right\rceil$\\
\qquad \textbf{while} $|\cX|>n_1$\\
\qquad\qquad  randomly choose $n_1$ inputs from $\cX$ and \emph{copy} them to $\cY$\\
\qquad\qquad  $\cX\leftarrow\algkosub(\cX)$\\
\qquad \textbf{end while} \\
\textbf{output:} $\algcomp(\cX \cup \cY)$
\end{algorithm}

In Theorem~\ref{thm:knockout}, we show that $\algkomod$ has $3$-approximation error less than $\errorcomb$. 

We first explain the algorithm, and then state the result. Let $n_1\df
\left\lceil\frac1\errorcomb\ln \frac1\errorcomb \cdot\log n
\right\rceil$. At each round, we add $n_1$ of the remaining inputs at
random to the multiset $\cY$ and run the knock-out subroutine
$\algkosub$ on the multiset $\cX$. When $|\cX| \le n_1$, we perform a
complete tournament on $\cX\cup\cY$, and declare the output as the
  winner. We show that, with probability at least $1-\errorcomb$, the
  final set $\cY$ contains at least one input which is a
  1-approximation of $\maxx$.
  probability greater than $1-\errorcomb$, an input within
  $1$-approximation of $\maxx$ remains in $\cX\cup\cY$. Since the
  complete tournament outputs a $2$-approximation of its maximum
  input, $\algkomod$ outputs a $3$-approximation of $\maxx$ with
  probability greater than $1-\errorcomb$.

\begin{theorem}
\label{thm:knockout}
For $n_1\ge 2$, we have $q_n^{\algkomod,\adaptive}<n+\frac12 (\log^4 n)\cdot \left\lceil\frac1\errorcomb\ln \frac1\errorcomb \right\rceil^2$ and $\pe_n^{\algkomod,\adaptive}(3)<\errorcomb$.
\end{theorem}
\begin{proof}
The number of comparisons made by $\algkosub$ is at most
$\frac{n}{2}+\frac{n}{4}+\frac{n}{8}+\ldots<n$. Observe that 
$\algkosub$ is called $m\df\left\lceil\log \frac{n}{n_1}\right\rceil$ times. 
Let $\cX_i$ be the multiset $\cX$ at the start of $i$th call to $\algkosub$.
Let $\cX_{m+1}$ and $\cY_{m+1}$ be the multisets $\cX$ and $\cY$ right
before calling $\algcomp$. Then, 
\begin{align*}
|\cX_{m+1}\cup\cY_{m+1}|&\le |\cX_{m+1}|+|\cY_{m+1}|\\
&\le n_1+ \sum_{i=1}^{m} \left(|\cY_{i+1}|-|\cY_{i}|\right)\\
&\le n_1+mn_1\\
&= \left(\left\lceil\log \frac{n}{n_1}\right\rceil+1\right) \cdot\left\lceil\frac1\errorcomb\ln \frac1\errorcomb \cdot \log n \right\rceil\\
&\le \left(\left\lceil\log \frac{n}{n_1}\right\rceil+1\right)\cdot \left\lceil\frac1\errorcomb\ln \frac1\errorcomb \right\rceil  \left\lceil\log n\right\rceil\\
&\le \log^2 n\cdot \left\lceil\frac1\errorcomb\ln \frac1\errorcomb \right\rceil,
\end{align*}
where the last inequality follows as $n_1\ge 2$ and $\log n$ is
an integer.  Since the complete tournament is quadratic in the
  input size, the total number of queries is at most $n+\frac12 \log^4
  n \left\lceil\frac1\errorcomb\ln \frac1\errorcomb \right\rceil^2$.
 
Next, we bound the error of $\algkomod$. Let
 \[
  \cX^*\df\{x\in\cX:x\ge\maxx-1\},
 \]
be the multiset of all inputs that are at least $\maxx-1$. For $i \le m+1$, let 
$\cX_i^*=\cX_i \cap \cX^*$ and $\cY_{m+1}^*=\cY_{m+1} \cap \cX^*$.
Let $\alpha_i\df\frac{|\cX_i^*|}{|\cX_i|}$ and $\alpha=\max\{\alpha_1,\alpha_2,\ldots,\alpha_m\}$.
We show that with high probability, $|\cX_{m+1}^*\cup\cY_{m+1}^*| \ge 1$, 
\ie some input in $\cX_{m+1}\cup\cY_{m+1}$ belongs to $\cX^*$.
In particular, we show that with probability $1-\errorcomb$, for large $\alpha$,
$|\cY^*_{m+1}| > 0$, and for small $\alpha$,
$\maxx\in \cX_{m+1}$. Observe that,
 \begin{align*}
 \pr(\maxx\notin\cX_{m+1}^*) &= \sum_{i=1}^m \pr(\maxx\notin\cX_{i+1}^*|\maxx \in \cX_i)\cdot\pr(\maxx \in \cX_i)\\
  &\le \sum_{i=1}^m \pr(\maxx\notin\cX_{i+1}^*|\maxx \in \cX_i)\\
  &\stackrel{(a)}{\le} \sum_{i=1}^m \frac{|\cX_i^*|-1}{|\cX_i|-1}\\
  &\le \sum_{i=1}^m \alpha_i\\
  &\le \alpha m,
 \end{align*}
where $(a)$ follows since at round $i$, $\algkosub$ randomly pairs the inputs and only inputs in
$\cX_{i}^*\backslash \{\maxx\}$ are able to eliminate $\maxx$. Next we discuss $\pr(|\cY_{m+1}^*|=0)$. At round $i$, the probability that an input in $\cX^*$ is not picked up in $\cY$ is
\[
\frac{\binom{|\cX_i|-|\cX_i^*|}{n_1}}{\binom{|\cX_i|}{n_1}}\le \left(1-\frac{|\cX_i^*|}{|\cX_i|}\right)^{n_1}=
\left(1-\alpha_i\right)^{n_1}.
\]
Therefore, 
 \begin{align*}
  \pr(|\cY_{m+1}^*|=0)&\le \prod_{i=1}^m (1-\alpha_i)^{n_1}\\
  &\le \min_i (1-\alpha_i)^{n_1}\\
  &=(1-\alpha)^{n_1}.
 \end{align*}
As a result, 
\begin{align*}
 \pr(|\cX_{m+1}^*\cup\cY_{m+1}^*|=0)&=\pr(|\cX_{m+1}^*|=0\textnormal{ }\land\textnormal{ } |\cY_{m+1}^*|=0)\\
 &\le \pr(\maxx\notin\cX_{m+1}^* \textnormal{ }\land\textnormal{ } |\cY_{m+1}^*|=0)\\
 &\le \max_\alpha \min \{\pr(\maxx\notin\cX_{m+1}^*),\pr(|\cY_{m+1}^*|=0)\}\\
 &\le \max_\alpha \min\left\lbrace \alpha m,(1-\alpha)^{n_1}\right\rbrace\\ 
 &\stackrel {(a)}{\le}   \left.\max\left\lbrace\alpha m,(1-\alpha)^{n_1}\right\rbrace\right|_{\alpha=\frac{\errorcomb}{\log n}}\\
 &=\max\left\lbrace\frac{\errorcomb m}{\log n},\left(1-\frac{\errorcomb}{\log n}\right)^{n_1}\right\rbrace\\
 &\stackrel {(b)}{<}\errorcomb,
\end{align*}
where $(a)$ follows since the first argument of the $\min$
increases and the second argument decreases with $\alpha$. Also, $(b)$
follows since $m\le \log n$ and $n_1 =\left\lceil\frac1\errorcomb\ln
\frac1\errorcomb \log n \right\rceil$.

So far, we have shown that with probability $1-\errorcomb$, there
exists a $1$-approximation of $\maxx$ in $\cX_{m+1}\cup\cY_{m+1}$.
From Lemma~\ref{lem:complete-tournament}, $\algcomp$ gives a
$2$-approximation of the maximum input. Consequently, with
  probability $1-\errorcomb$, $\algkomod$ outputs a $3$-approximation
  of $\maxx$.
\end{proof}

In Appendix~\ref{sec:app-example}, we show that $\algkomod$ cannot
output better than $3$-approximation of $\maxx$ with constant probability. 

We end this subsection with the following open question.
\begin{openq}
What is the best approximation that the simple knock-out algorithm can achieve?
\end{openq}

\subsection{Quick-select}
\label{subsec:quickselect}
Motivated by quick-sort, we propose a quick-select algorithm $\algquickselect$ that at each round
compares all the inputs with a random
pivot to provide stronger performance guarantees against the non-adaptive adversary. 

\begin{algorithm}
\caption{$\algqssub$ - Subroutine for $\algquickselect$ and $\algcomb$}
\textbf{input:} $\cX$\\
\qquad pick a pivot $x_p\in\cX$ at random\\
\qquad compare $x_p$ with all other inputs in $\cX$\\
\qquad 	let $\cY\subset \cX\backslash\{x_p\}$ be the multiset of inputs that beat $x_p$\\
\textbf{output:} if $\cY\ne \emptyset$ output $\cY$ otherwise output $\{x_p\}$ 
\end{algorithm}
\begin{algorithm}
\caption{$\algquickselect$ - Quick-select}
\textbf{input:} $\cX$\\
\qquad \textbf{while} $|\cX|>1$\\
\qquad\qquad $\cX\leftarrow\algqssub(\cX)$\\
\qquad \textbf{end while}\\
\textbf{output:} the unique input in $\cX$
\end{algorithm}

We show that $\algquickselect$ provides a $2$-approximation with no
error against both the adaptive and non-adaptive adversaries. To show this result, observe that $\maxx$ will only be eliminated if a $1$-approximation of $\maxx$ is chosen as pivot and therefore, only inputs that are $2$-approximation of $\maxx$ will survive.

\begin{lemma}
\label{lem:quick}
  $\pe_n^{\algquickselect,\adaptive}(2)=0$.
\end{lemma}
\begin{proof}
 If the output is $\maxx$, the lemma holds. Otherwise, $\maxx$
 is discarded when it was chosen as a pivot or 
 compared with a pivot. Let $x_p$ be the pivot when $\maxx$ 
 is discarded, hence $x_p\ge \maxx-1$. 
 By the algorithm's definition, all the surviving inputs
 are at least $x_p-1\ge \maxx-2$.
\end{proof}

We now show that the expected query complexity of
  $\algquickselect$ against a non-adaptive adversary is at most
  $2n$. This result follows from the observation that the non-adaptive
  adversary fixes the comparison graph from the start, and
    hence a random pivot wins against half of the inputs in
  expectation. This idea is made rigorous in the proof of
  Lemma~\ref{lem:expected_quick_select}.

We finally consider an example for which $\algquickselect$ requires
$\binom{n}{2}$ queries against the adaptive adversary.  
\begin{lemma}
 \label{lem:expected_quick_select}
 $q_n^{\algquickselect,\nonadaptive} < 2n$.
\end{lemma}
\begin{proof}
Recall that the non-adaptive adversary can be modeled as a complete
directed graph where each node is an input and there is an edge from
$x$ to $y$ if $\cC(x,y)=x$. Let $\din(x)$ be the in-degree of $x$.

At round $i$ the algorithm chooses a pivot $x_p$ at random and
compares it to all the remaining inputs. By keeping the winners,
$\max\{\din(x_p),1\}$ inputs will remain for the next round. As a
result, we have the following recursion for non-adaptive adversaries,
\begin{align*}
 q_n^{\algquickselect}&=\EE\left[Q_n^{\algquickselect}\right]\\
 &= n-1+\frac{1}{n}\sum_{i=1}^{n} \EE\left[Q_{\din(x_i)}^{\algquickselect}\right]\\
 &=n-1+\frac{1}{n}\sum_{i=1}^{n} q_{\din(x_i)}^{\algquickselect}.
\end{align*}
By~\eqref{eqn:nonadaptivemodel},
\begin{align}
 \label{eqn:lem_expected}
  q_n^{\algquickselect,\nonadaptive}&=\max_{\cC\in\nonadaptivecomp}\max_{\cX} q_n^{\algquickselect}\\
 \nonumber
&=\max_{\cC\in\nonadaptivecomp}\max_{\cX} \left[n-1+\frac{1}{n}\sum_{i=1}^{n} q_{\din(x_i)}^{\algquickselect}\right]\\
 \nonumber
 &\le n-1+\frac{1}{n}\sum_{i=1}^{n} \max_{\cC\in\nonadaptivecomp}\max_{\cX} q_{\din(x_i)}^{\algquickselect}\\
 \nonumber
 &=n-1+\frac{1}{n}\sum_{i=1}^{n} q_{\din(x_i)}^{\algquickselect,\nonadaptive},
\end{align}
where the inequality follows as maximum of sums is at most sum of maximums. We prove by strong induction that
$q_n^{\algquickselect,\nonadaptive}\le 2(n-1)$. It holds for $n=1$. Suppose it holds for all
$n'<n$, then,
 \begin{align*}
  q_n^{\algquickselect,\nonadaptive}&\le n-1+\frac{1}{n}\sum_{i=1}^{n} q_{\din(x_i)}^{\algquickselect,\nonadaptive}\\
  &\le n-1+\frac{1}{n} \sum_{i=1}^{n} 2\cdot \din(x_i)\\
  &=n-1+\frac{n(n-1)}{n}\\
  &\le2(n-1),
 \end{align*}
where the equality follows since the in-degrees sum to $\frac{n(n-1)}{2}$.
\end{proof}

Lemma~\ref{lem:expected_quick_select} shows that
$q_n^{\algquickselect,\nonadaptive}<2n$.  Next, we show a naive concentration bound for the query complexity of $\algquickselect$. By Markov's inequality, for a non-adaptive adversary,
\[
 \pr(Q_n^{\algquickselect}>4n)\le \frac12.
\]
Let $k$ be an integer multiple of $4$. Now suppose we run $\algquickselect$, allowing $kn$ queries. At each $4n$ queries, the $\algquickselect$ ends with probability $\ge \frac12$. Therefore,
\[
\pr(Q_n^{\algquickselect}>kn)\le 2^{-\frac{k}{4}}.
\]
This naive bound is exponential in $k$. The next lemma shows
  a tighter super-exponential concentration bound on the query complexity of the algorithm
  beyond its expectation. We defer the proof to
  appendix~\ref{sec:app-upp}.
\begin{lemma}
\label{lem:conselect}
 Let  $k'=\max\{e,k/2\}$. For a non-adaptive adversary, $\pr(Q_n^{\algquickselect}>kn)\le e^{-(k-k')\ln k'}$.
\end{lemma}

While $\algquickselect$ has linear expected query complexity under the
non-adaptive adversarial model, the following example suggested to us
by Jelani Nelson~\cite{JelaniNPC} shows that it has a quadratic query
complexity against an adaptive adversary.
\begin{example}
\label{exm:expected_quick_select}
 Let $\cX=\{0,0,\ldots, 0\}$. At each round, the adversary declares the pivot to be smaller than all the other inputs. Consequently, only the pivot is eliminated and the query complexity is $\binom{n}{2}$. 
\end{example}

\subsection{Knock-out and quick-select combination}
\label{sec:combined}
$\algkomod$ has the benefit of reducing the
  number of inputs exponentially at each round and therefore
  maintaining a linear query-complexity while having only a
  $3$-approximation guarantee. On the other side, $\algquickselect$ has a
  $2$-approximation guarantee while it may require $\cO(n^2)$ queries
  for some instances of inputs. In $\algcomb$ we combine the benefits
  of these algorithms and avoid their shortcomings. By carefully
  repeating $\algqssub$ we try to reduce the number of inputs by a
  fraction at each round and keep the largest element in the remained
  set. If the number of inputs is not reduced by a fraction, most of
  them must be close to each other, therefore repeating the
  $\algkosub$ for a sufficient number of times and 
  keeping the inputs with higher
  number of wins will guarantee the reduction of the input size
  without adversing the approximation error.
Our final algorithm $\algcomb$ provides a $2$-approximation of $\maxx$
even against the adaptive-adversarial comparator, and has linear query
complexity, therefore resolves an open question
of~\cite{AjtaiFHN15}.

\begin{algorithm}
\label{alg:qks}
\caption{$\algcomb$ - Knock-out and quick-select combination}
\textbf{input:} $\cX,\errorcomb$\\
\qquad $\beta_1=9$, $\beta_2=25$, $i=0$\\
\qquad \textbf{while} $|\cX|>1$	\\
\qquad\qquad $i=i+1$ \qquad\qquad($i$ is the round)\\
\qquad\qquad $n_i=|\cX|$\\
\qquad\qquad run $\cX\leftarrow \algqssub(\cX)$ for $\left\lfloor\beta_1 \log\tfrac1\errorcomb\right\rfloor$ times\\
\qquad\qquad $\cX_i=\cX$\\
\qquad\qquad \textbf{if} $|\cX_i|>\frac23 n_i$\\
\qquad\qquad\qquad run $\algkosub$ on \emph{fixed} $\cX$ for $\left\lfloor\beta_2 \left(\frac43\right)^i \log\tfrac1\errorcomb\right\rfloor$ times\\
\qquad\qquad\qquad \textbf{if} there exists an input with $>\frac34 \left\lfloor\beta_2 \left(\frac43\right)^i \log\tfrac1\errorcomb\right\rfloor$ wins\\
\qquad\qquad\qquad\qquad let $\cX$ be a multiset of inputs with $>\frac34 \left\lfloor\beta_2 \left(\frac43\right)^i \log\tfrac1\errorcomb\right\rfloor$ wins\\
\qquad\qquad\qquad \textbf{else} \\
\qquad\qquad\qquad\qquad let $\cX$ be an input with highest number of wins\\
\qquad \textbf {end while}\\
\textbf{output:} $\cX$
\end{algorithm}

We begin analysis of the algorithm with a few simple lemmas. 
\begin{lemma}
\label{lem:size-reduction}
At each round $|\cX|$ reduces at least by a third, \ie $n_{i+1}\le \frac23 n_i$.
\end{lemma}
\begin{proof}
If at any round $|\cX_i|\le \frac{2}{3} n_i$, then the lemma holds and the algorithm does not call $\algkosub$. On the other hand, if $\algkosub$ is called, then by Markov's inequality at most $\frac23$rd of the inputs win more than $\frac34$th fraction of queries. As a result, at round $i$ at least $\frac13$rd of the inputs in $\cX$ will be eliminated.
\end{proof}
Recall that $\cX^*=\{x\in\cX:x\ge\maxx-1\}.$ The next lemma shows that
choosing inputs inside $\cX^*$ as a pivot, guarantees a
$2$-approximation of $\maxx$. The proof is similar to
Lemma~\ref{lem:quick} and is omitted.
\begin{lemma}
\label{lem:calltopivot}
If $\maxx\in \cX$, at a call to $\algqssub$ either $\maxx$ survives or a pivot from $\cX^*$ is chosen where in the later case, 
only inputs that are $2$-approximation of $\maxx$ will survive.
\end{lemma}

We showed that at each round, $\algcomb$ reduces $|\cX|$ by at least a third. As a result, the number of inputs decreases exponentially and the total number of queries is linear in $n$.  We also show that if $\maxx$ is eliminated at some round, then at that round, an input from $\cX^*$ has been chosen as a pivot with high probability. Using Lemma~\ref{lem:calltopivot}, this implies that $\algcomb$ outputs a $2$-approximation of $\maxx$ with high probability.

\begin{theorem}
\label{thm:qks}
$q_n^{\algcomb,\adaptive}=\cO\left(n\log \tfrac1\errorcomb\right)$ and $\pe_n^{\algcomb,\adaptive}(2) < \errorcomb$.
\end{theorem}
\begin{proof}
We start by analyzing the query complexity of~$\algcomb$. By Lemma~\ref{lem:size-reduction},
\[
n_i\le n\cdot \left(\tfrac{2}{3}\right)^{i-1}.
\]
Therefore, the total number of queries at round $i$ is at most
\begin{align*}
n\left(\tfrac{2}{3}\right)^{i-1} \beta_1\log\tfrac1\errorcomb + \tfrac{n}{2} \left(\tfrac23\right)^{i-1}\beta_2 \left(\tfrac43\right)^i \log\tfrac1\errorcomb,
\end{align*}
where the first term is for calls to~$\algqssub$ and the second
term is for calls to~$\algkosub$. Adding the query complexity of all
the rounds,
\begin{align*}
 q_n^{\algcomb,\adaptive}&\le n\log\tfrac1\errorcomb\sum_{i=1}^\infty  \left(\beta_1\left(\tfrac{2}{3}\right)^{i-1}+\tfrac23 \beta_2 \left(\tfrac{8}{9}\right)^{i-1}\right)\\
 &\le n(3\beta_1+6\beta_2) \log\tfrac1\errorcomb\\
 &=\cO\left(n\log\tfrac1\errorcomb\right).
\end{align*}

We now analyze the approximation guarantee of~$\algcomb$.  We show that at least one of the following events happens with probability greater than $1-\errorcomb$.
\begin{itemize}
\item $\algcomb$ outputs $\maxx$.
\item An input inside $\cX^*$ is chosen as a pivot at some round.  
\end{itemize}

Let $\cX_i^*\df\cX_i \cap \cX^*$ and $\alpha_i\df\frac{|\cX_i^*|}{|\cX_i|}$.
We consider the following two cases separately.
\begin{itemize}
 \item \textbf{Case 1} There exists an $i$ such that $|\cX_i|>\frac23 n_i$ and $\alpha_i > \frac18$.
 \item \textbf{Case 2} For all $i$, either  $|\cX_i|\le\frac23 n_i$ or $\alpha_i \le \frac18$.
\end{itemize}

First we consider case 1. We show that in this case a pivot from
$\cX^*$ is chosen with probability $> 1-\errorcomb$.  Observe that at
round $i$, $|\cX|$ starts at $n_i < \frac{3}{2}|\cX_i|$ and gradually
decreases. On the other hand, in all the $\left\lfloor\beta_1
\log\tfrac1\errorcomb\right\rfloor$ calls to $\algqssub$,
$|\cX\cap\cX^*|$ is at least $|\cX_i^*|=\alpha_i|\cX_i|$. Therefore,
in all the calls to $\algqssub$ at round $i$,
\[
 \frac{|\cX\cap\cX^*|}{|\cX|}\ge
 \frac{\alpha_i|\cX_i|}{\frac{3}{2}|\cX_i|}=\frac23\alpha_i.
\]
Let $E$ be the event of not choosing a pivot from $\cX^*$ at round
$i$. As a result,
\begin{align*}
 \pr(E)&\le \left(1-\tfrac23\alpha_i\right)^{\left\lfloor\beta_1\log\tfrac1\errorcomb\right\rfloor}\\
 &\le \left(\tfrac{11}{12}\right)^{8\log\tfrac1\errorcomb}\\
 &<\errorcomb.
\end{align*}
Therefore, in case 1, with probability at least $1-\errorcomb$, a
pivot from $\cX^*$ is chosen.

We now consider the case 2. By Lemma~\ref{lem:calltopivot} during the
calls to $\algqssub$, either $\maxx$ survives or an input from $\cX^*$
is chosen as a pivot.  Therefore, we may only lose $\maxx$ without
choosing a pivot from $\cX^*$, if at some round $i$, $|\cX_i|>\frac23
n_i$ and $\maxx$ wins less than $\frac34$th of its queries during the
calls to $\algkosub$.

Recall that in case 2, if $|\cX_i|>\frac23 n_i$ then $\alpha_i\le
\frac18$. Observe that $\maxx$ wins against a random input in $\cX_i$
with probability greater than $>1-\alpha_i$ which is at least
$\frac78$. Let $E'_i$ be the event that $\maxx$ wins fewer than
$\frac34$th of its queries at round $i$. By the Chernoff bound,
\begin{align*}
 \pr(E'_i)&\le \exp\left(-\left\lfloor\beta_2 \left(\tfrac43\right)^i \log\tfrac{1}{\errorcomb}\right\rfloor\cdot D\left(\tfrac34||\tfrac78\right) \right)\\
 &\le \errorcomb^{2\left(\tfrac43\right)^i},
\end{align*}
where $D(p||q) \df p \ln \frac{p}{q} + (1-p) \ln \frac{1-p}{1-q}$ is
the Kullback-Leibler distance between Bernoulli distributed random
variables with parameters $p$ and $q$ respectively. Assuming
$\errorcomb<\frac12$, the total probability of missing $\maxx$ without
choosing a pivot form $\cX^*$ is at most
\begin{align*}
 \sum_{i=1}^{\infty} \Pr(E'_i) &\le \sum_{i=1}^{\infty} \errorcomb^{2\left(\tfrac43\right)^i}\\ 
 &< \errorcomb.
\end{align*}

So far we showed that with probability $> 1-\errorcomb$, either
$\maxx$ survives or an input inside $\cX^*$ is chosen as a pivot. The
theorem follows from Lemma~\ref{lem:calltopivot}.
\end{proof}

\section{Application to density estimation}
\label{sec:densityestimation}

Our study of maximum selection with adversarial comparators was
motivated by the following density estimation problem.

Given a \emph{known} set $\cP_{\deltacover}=\{p_1,\ldots,p_n\}$ of $n$
distributions and $k$ samples from an \emph{unknown} distribution
$p_0$, output a distribution $\hat{p} \in \cP_{\deltacover}$ such that
for a small constant $C>1$ and with high probability,

\[
\|\hat{p}-p_0\|_1\le C\cdot \min_{p \in \cP_{\deltacover}} \|p-p_0\|_1+o_k(1).
\]

\noindent This problem was studied in~\cite{DevroyeL01} who showed
that for $n=2$, the $\algscheffetest$, described below in pseudocode, takes
$k$ samples and with probability $1-\errordens$ outputs a distribution
$\hat{p}\in \cP_{\deltacover}$ such that
 \begin{equation}
 \label{eqn:scheffetest}
  \norm{\hat{p}-p_0}_1 \leq 3 \cdot \min_{p \in \cP_{\deltacover}} \norm{p-p_0}_1 + \sqrt{\frac{10\log \frac{1}{\errordens}}{k}}.
 \end{equation}
\begin{algorithm}
\caption{$\algscheffetest$- Scheffe test for two distributions}
\textbf{input:} distributions $p_1$ and $p_2$, $k$ $\iid$ samples of unknown distribution $p_0$\\
\qquad let $\cS=\{x:p_1(x)>p_2(x)\}$\\
\qquad let $p_1(\cS)$ and $p_2(\cS)$ be the probability mass that $p_1$ and $p_2$ assign to $\cS$\\ 
\qquad let $\mu_\cS$ be the frequency of samples in $\cS$\\

\textbf{output:} if $|p_1(\cS)-\mu_\cS|\le|p_2(\cS)-\mu_\cS|$ output $p_1$, otherwise output $p_2$ 
\end{algorithm}

$\algscheffetest$ provides a factor-$3$ approximation with high
probability.  The algorithm, as stated in its pseudocode, requires
computing $p_i(\cS)$ which can be hard since the distributions are not
restricted.  However, as noted in~\cite{AcharyaJOS14}, the algorithm
can be made to run in time linear in $k$. 
the probability of the samples under the known distributions.
\cite{DevroyeL01} also extended $\algscheffetest$ for $n>2$. Their
proposed algorithm for $n>2$, runs $\algscheffetest$ for each pair of
distributions in $\cP_{\deltacover}$ and outputs the distribution with
maximum number of wins, where a distribution is a winner if it is the
output of $\algscheffetest$. This algorithm is referred to as the
Scheffe tournament.  They showed that this algorithm finds a
distribution $\hat{p}\in \cP_{\deltacover}$ such that
\[
 \norm{\hat{p}-p_0}_1\le 9 \min_{p \in \cP_{\deltacover}} \norm{p-p_0}_1+o_k(1), 
\]
and the running time is clearly $\Theta(n^2k)$, quadratic in the number
of distributions. 

\cite{MahalanabisS08} showed that the optimal coefficients for the
Scheffe algorithms are indeed 3 and 9 for $n=2$ and $n>2$
respectively.
They proposed an algorithm with an improved factor-$3$ approximation for
$n>2$, however still running in time $\Theta(n^2)$. They
also proposed a linear-time algorithm, but it requires a
preprocessing step that runs in time exponential in $n$. 

Scheffe's method has been used recently to
obtain nearly sample optimal algorithms for learning Poisson Binomial
distributions~\cite{DaskalakisDS12}, and
Gaussian mixtures~\cite{DaskalakisK13, AcharyaJOS14}. 

We now describe how our noisy comparison model can be applied to this
problem to yield a linear-time algorithm with the same estimation
guarantee as Scheffe tournament. Our algorithm uses Scheffe test as a
subroutine. Given a sufficient number of samples, $k=\Theta(\log n)$,
the small term in the RHS of~\eqref{eqn:scheffetest} vanishes and
$\algscheffetest$ outputs
\[
\left\{
\begin{array}{ll}
p_i & \text{if } \norm{p_i-p_0}_1 <\frac13 \norm{p_j-p_0}_1,\\
p_j & \text{if } \norm{p_j-p_0}_1 <\frac13 \norm{p_i-p_0}_1,\\
\textnormal{unknown} & \text{otherwise.} 
\end{array} \right.
\]
Let $x_i=-\log_3\norm{p_i-p_0}_1$, then analogously to the maximum
selection with adversarial noise in~\eqref{eq:comp}, $\algscheffetest$
outputs
\[
\left\{
\begin{array}{ll}
\max\{x_i,x_j\} & \text{if } |x_i-x_j|>1,\\
\textnormal{unknown} & \text{otherwise.} 
\end{array} \right.
\]
Given a fixed multiset of samples the tournament results are fixed,
hence this setup is identical to the non-adaptive adversarial
comparators.  In particular, with probability $1-\errordens$, our
quick-select algorithm can find $\hat{p}\in \cP_{\deltacover}$ such
that
\[
 \norm{\hat{p}-p_0}_1\le 9\cdot \min_{p \in \cP_{\deltacover}} \norm{p-p_0}_1,
\]
with running time $\Theta(nk)$. Next, we consider the combination of $\algscheffetest$ and $\algquickselect$ in more details.
\begin{theorem}
 Combination of $\algscheffetest$ and $\algquickselect$ algorithms,
 with probability $1-\errordens$, results in $\hat{p}$ such that
\[
 \norm{\hat{p}-p_0}_1\le 9\cdot \min_{p \in \cP_{\deltacover}}
 \norm{p-p_0}_1+4\sqrt{\frac{10\log \frac{\binom n2}{\errordens}}{k}}.
\]
\end{theorem}
\begin{proof}
Let
\[
 p^*\df\argmin_{p\in\cP_\deltacover}\norm{p-p_0}_1.
\]
Using~\eqref{eqn:scheffetest}, for each $p_i$ and $p_j$ in
$\cP_\deltacover$, with probability $1-\errordens/\binom n2$,
$\algscheffetest$ outputs $\hat{p}$ such that
\begin{align}
\label{eqn:ttt}
\norm{\hat{p}-p_0}_1 \le 3 \cdot \min_{p \in \{p_i,p_j\}}
\norm{p-p_0}_1 + \sqrt{\frac{10\log \frac{\binom n2}{\errordens}}{k}}.
\end{align}
By the union bound~\eqref{eqn:ttt} holds for all $p_i$ and $p_j$
with probability at least $1-\errordens$. Similar to
Lemma~\ref{lem:quick}, if $p^*$ is eliminated, then at some round,
$\algquickselect$ has chosen $p'$ as a pivot such that
\[
\norm{p'-p_0}_1 \le 3 \cdot \norm{p^*-p_0}_1 + \sqrt{\frac{10\log \frac{\binom n2}{\errordens}}{k}}.
\]
Now after choosing $p'$ as a pivot, for any distribution $p''$ that survives, 
\begin{align*}
 \norm{p''-p_0}_1 &\le  3 \cdot \norm{p'-p_0}_1 + \sqrt{\frac{10\log \frac{\binom n2}{\errordens}}{k}} \\
 &\le 9 \cdot \norm{p^*-p_0}_1 + 4 \sqrt{\frac{10\log \frac{\binom n2}{\errordens}}{k}}.
\end{align*}
\end{proof}

\section{Noisy sorting}
\label{sec:sorting}
\subsection{Problem statement}
We now consider sorting with noisy comparators. The comparator model
is the same as before, and the goal is to approximately sort the
inputs in decreasing order.

Consider an Algorithm $\cA$ for sorting the inputs. The output of
$\cA$ is denoted by $\bY_\cA(\cX)\df(Y_1,Y_2,\ldots,Y_n)$ which is a
particular ordering of the inputs.
Similar to the maximum-selection problem, a
$t$-approximation error is
\[
 \pe_{n}^{\cA}(t) \df \pr\left(\max_{i,j:i>j}(Y_i-Y_j) > t\right),
\]
$\ie$ the probability of $Y_i$ appearing after $Y_j$ in $\bY_\cA$
while $Y_i-Y_j > t$. Note that our definitions for
$\pe_n^{\cA,\nonadaptive}(t)$, $\pe_n^{\cA,\adaptive}(t)$,
$q_n^{\cA,\adaptive}$, and $q_n^{\cA,\nonadaptive}$ hold same as
before.

In the following, we first revisit complete tournament with a small
modification for the sake of sorting problem and we show that under
adaptive adversarial model, it has zero $2$-approximation error and
query complexity of $\binom{n}{2}$. Then we discuss quick-sort
algorithm $\algquicksort$ and show that it has zero $2$-approximation
error but with improved query complexity for the non-adaptive
adversary. We apply the known bounds for running time of general
quick-sort algorithm with $n$ distinct inputs to find the query
complexity of $\algquicksort$.

\subsection{Complete tournament}
The algorithm is similar to $\algcomp$ in
Section~\ref{subsec:complete} and we refer to it as
$\algcompsort$. The only difference is in the output of the algorithm.
\begin{algorithm}
\caption{$\algcompsort$ - Complete tournament} \textbf{input:}
$\cX$\\ \qquad compare all input pairs in $\cX$, count the number of
times each input wins\\ \textbf{output:} output the inputs in the
order of their number of wins, breaking the ties randomly
\end{algorithm}

The following lemma and its proof is similar to
Lemma~\ref{lem:complete-tournament} and therefore we skip the proof.
\begin{lemma}
 \label{lem:complete-tournamentsort}
  $q_n^{\algcompsort,\adaptive}=\binom{n}{2}$ and $\pe_n^{\algcompsort,\adaptive}(2)=0$.
\end{lemma}
Next, we discuss an algorithm with improved query complexity.

\subsection{Quick-sort}
Quick-sort is a well known algorithm and here is denoted by
$\algquicksort$. The expected query complexity of quick-sort with
noiseless comparisons and distinct inputs is
\begin{equation}
\label{eqn:quicksortfinal}
  f(n)\df2n\ln n-(4-2\gamma)n+2\ln n+\cO(1),
\end{equation}
where $\gamma$ is Euler's constant \cite{McdiarmidH96}. Note that
$f(n)$ is a convex function of $n$.

In the rest of this section we study the error guarantee of quick-sort
and its query complexity in the presence of noise. In
Lemma~\ref{lem:quicksort} we show that the error guarantee of
quick-sort for our noise model is same as complete tournament, $\ie$
it can sort the inputs with zero $2$-approximation error. Next in
Lemma~\ref{lem:expectedquicksort} we show that the expected query
complexity of quick-sort with non-adaptive adversarial noise is at
most its expected query complexity in noiseless model.

\begin{lemma}
 \label{lem:quicksort}
  $\pe_n^{\algquicksort,\adaptive}(2)=0$.
\end{lemma}
\begin{proof}
 The proof is by contradiction. Suppose $x_i>x_j+2$ but $x_j$ appears
 before $x_i$ in the output of quick-sort algorithm. Then there must
 have been a pivot $x_p$ such that $\cC(x_i,x_p)=x_p$ while
 $\cC(x_j,x_p)=x_j$. Since $x_i>x_j+2$ no such a pivot exists.
\end{proof}

Quick-sort algorithm chooses a pivot randomly to divide the set of
inputs into smaller-size sets. The optimal pivot for noiseless
quick-sort is known to be the median of the inputs to balance the size
of the remained sets. In fact, it is easy to show that if we choose
the median of the inputs as pivot, the query complexity of quick-sort
reduces to less than $n\log n$. Observe
that in a non-adaptive adversarial model, the probability of having
balanced sets after choosing pivot increases. As a result, in the next
lemma we show that the expected query complexity of quick-sort in the
presence of noise is upper bounded by $f(n)$.
\begin{lemma}
\label{lem:expectedquicksort}
 $q_n^{\algquicksort,\nonadaptive}=f(n)$ and is achieved when
the queries are noiseless and inputs are distinct.
\end{lemma}
\begin{proof}
 Let $\din(x)$ and $\dout(x)$ be the in-degree and out-degree of node
 $x$ in the complete tournament respectively. For the noiseless
 comparator with distinct inputs, the in-degrees and out-degrees of
 inputs are permutation of $(0,1,\ldots,n-1)$. We show that
 \[
   \argmax_{\cC\in\nonadaptivecomp}\max_{\cX} q_n^{\algquicksort},
 \]
 is a comparator whose complete tournament in-degrees and
 out-degrees are permutations of $(0,1,\ldots,n-1)$. For the simplicity
 of notation let $q_n=q_n^{\algquicksort,\nonadaptive}$. We have the
 following recursion for quick-sort similar
 to~\eqref{eqn:lem_expected}.
 \begin{align}
 \label{eqn:quicksort}
 q_n\le n-1+\frac{1}{n}\sum_{i=1}^n q_{\dout(x_i)}+q_{\din(x_i)}
\end{align}
By induction, we show that the solution to~\eqref{eqn:quicksort} is
bounded above by $f(n)$, a convex function of $n$. The induction holds
for $n=0,1,$ and $2$. Now suppose the induction holds for all
$i<n$. Since $f(n)$ is a convex function of $n$ and $\sum_i
\din(x_i)=\sum_i \dout(x_i)=\frac{n(n-1)}{2}$, the right hand side
of~\eqref{eqn:quicksort} is maximized when the in-degrees and
out-degrees take their extreme values, $\ie$ when they are permutation
of $(0,1,\ldots,n-1)$. Plugging in these values,~\eqref{eqn:quicksort}
is equivalent to,
\begin{align*}
 q_n&\le n-1+\frac{1}{n}\sum_{i=1}^n f(\din(x_i))+f(\dout(x_i))\\
    &\le n-1+\frac{1}{n}\sum_{i=1}^n f(i-1)+f(n-i),
\end{align*}
where the solution to this recursion is $f(n)$, given
in~\eqref{eqn:quicksortfinal}.  Hence $q_n$ is bounded above by $f(n)$
and the equality happens when the in-degrees and out-degrees are
permutations of $(0,1,\ldots,n-1)$.
\end{proof}

\cite{Donald99,Hennequin89,McdiarmidH92} show different concentration
bounds for quick-sort. In particular, \cite{McdiarmidH92} show that
the probability of quick-sort algorithm requiring more comparisons
than $(1+\errorquicksort)$ times its expected query complexity is
$n^{-2\errorquicksort\ln\ln n +\cO(\ln \ln \ln n)}$. Observe that for
the non-adaptive adversarial model, the chance of a random pivot
cutting the set of inputs into balanced sets increases. As a result,
one can show that the analysis in \cite{McdiarmidH92} follows
automatically. In particular, Lemmas~2.1 and~2.2 in
\cite{McdiarmidH92}, which are the basis of their analysis, are valid
for our non-adaptive adversarial model. Therefore, their tight
concentration bound for quick-sort algorithm can be applied to our
non-adaptive adversarial model. 
\paragraph{Acknowledgment}
We thank Jelani Nelson for introducing us to the problem's adaptive
adversarial model.

\bibliographystyle{alpha}
\bibliography{abr,masterref}

\appendix
\pagebreak
\section{For all $t<3$, $\algkomod$ cannot output a $t$-approximation}
\label{sec:app-example}
The next example shows that the modified knock-out algorithm cannot achieve
better than $3$-approximation of $\maxx$. 
\begin{example}
\label{exm:knockouterror}
Suppose $n-2$ is multiple of $3$ and $n$ is a large number. Let $\cX$ be a random permutation of
\[
 \{3,\underbrace{2,2,\ldots,2}_{\frac{n-2}{3}},\underbrace{1,1,\ldots,1}_{\frac{n-2}{3}},\underbrace{0,0,\ldots,0}_{\frac{n-2}{3}},0^*\}.
\]
This multiset consists of an input with value zero but specified with
$0^*$, since this input is going to behave differently from other
$0$s. Let the adversarial comparator be such that all $0$s, except
$0^*$, and all $2$s lose to all $1$s, and $3$ loses to all $2$s. By
the properties of comparator it is obvious that any $2$ will defeat
all zeros, including $0^*$. In order to prove our main claim, we make
the following arguments and show that each of them happens with high
probability.
\begin{itemize}
\item $\pr$(input with value $3$ is not present in the final multiset)$>\frac3{10}$.
\item  $\pr$(input $0^*$ is present in the final multiset)$>\frac13$.
\item With high probability, the fraction of $1$s in the final multiset is close to $1$.
\end{itemize}
Before proving each argument we show why all the above statements are sufficient to prove our claim. Consider the final multiset; with high probability it mainly consists of $1$s and there are small number of $0$s and $2$s. Moreover, with probability greater than $\frac1{3}\times \frac3{10}$, input with value $3$ has been removed before reaching the final multiset and $0^*$ has survived to reach the final multiset. Therefore, if we run Algorithm $\algcomp$ on the final multiset, the input $0^*$ will have the most number of wins and declared as the output of the algorithm. Hence for all $t<3$, we have $\pe_n^{\algkomod,\nonadaptive}(t)>\text{constant}$. Note that we did not try to optimize this constant.

Now we show why each of the arguments above are true. Note that all the discussions made below is in expected value. However, the concentration bounds for all these claims are straightforward but not discussed because it is out of the scope of this paper. Also we assume that $n$ is sufficiently large. 
\begin{lemma}
\label{lem:1frac}
With high probability, the fraction of $1$s in the final multiset is close to 1 and the fraction of $0$s and $2$s are very small.
\end{lemma}
\begin{proof}
We calculate the expected number of $0$s, $1$s, and $2$s at each step. Let $f_i(j)$ be the fraction of $j$'s at the end of step $i$. After each step, we lose an input with value $1$ if and only if they are paired with each other. So we have the following recursion,
\[
f_{i+1}(1)=2 \cdot f_i(1) \left(\tfrac{f_i(1)}{2}+1-f_i(1)\right),
\]
where the factor $2$ on the RHS of the recursion above is due to the
fact that at each step we are reducing the number of inputs to
half. Starting with $f_0(1)=1/3$, we get the set of values $\{1/3,
5/9, 65/81,6305/6561\sim 0.96, ... \}$ for $f_i(1)$s. We can see that
the ratio is approaching 1 very fast. More precisely, the fraction of
$0$s is decreasing quadratically since their only chance of survival
is to get paired among themselves. As a result, after a couple of
steps, the fraction of zeros is extremely small, and henceforth the
only chance of survival for $2$s becomes getting paired among
themselves and also their fraction is going to decrease quadratically
afterwards. As a result, more samples of $1$s will be in the final
$\cY$ with high probability.
\end{proof}
\begin{lemma}
\label{lem:rmv3}
\pr(input with value $3$ is not present in the final multiset)$>\frac3{10}.$
\end{lemma}
\begin{proof}
The input with value $3$ is going to be removed when it is compared
against one of the $2$s. There is a slight chance of surviving for it
if it is chosen randomly for being in the output. Thus the probability
of input $3$ being removed from the multiset in the first round is
\[
\pr\text{(input } 3 \text{ is being removed in the first round)}=\frac{n-2}{3n}\left(1-\frac{n_1}{n}\right) > \frac3{10},
\]
where $n_1 =\left\lceil\frac1\errorcomb\ln \frac1\errorcomb \log n \right\rceil$.
\end{proof}
\begin{lemma}
\label{lem:nrmv0*}
\pr(input $0^*$ is present in the final multiset)$>\frac13.$
\end{lemma}
\begin{proof}
Similar to the argument made in the proof of Lemma~\ref{lem:1frac}, we have the following recursion for $f_i(2)$.
\[
f_{i+1}(2)=2 \cdot f_i(2) \left(\tfrac{f_i(2)}{2}+1-f_i(2)-f_i(1)\right)
\]
Thus we have $f_0(2)=1/3$, $f_1(2)=1/3$, $f_2(2)=5/27$,
$f_3(2)=85/2187$. As we stated in the proof of Lemma~\ref{lem:1frac},
the expected fraction of $2$s is decreasing quadratically and
\[
 \pr(0^*
 \text{surviving})=(1-\tfrac13)(1-\tfrac13)(1-\tfrac5{27})(1-\tfrac{85}{2187})\cdots
 >\frac13,
\]
proving the lemma. 
\end{proof}
\end{example}

\section{Proof of Lemma~\ref{lem:conselect}}
\label{sec:app-upp}
Abbreviate $Q_n^{\algquickselect}$ by $Q_n$. As in the Chernoff bound proof, for all $\lambda>0$,
 \begin{align}
  \label{eqn:chernoff}
  \pr(Q_n>kn) \le \frac{\EE[e^{\lambda Q_n}]}{e^{k\lambda n}}.
 \end{align}
 Let $\lambda=\frac1n \ln k'$ and $\Phi(i)\df\EE[e^{\lambda Q_i}]$. We
 prove by induction that $\Phi(i)\le e^{k'\lambda i}$. The induction
 holds for $i=0$. Similar to~\eqref{eqn:lem_expected}, we have the
 following recursion for $\Phi(n)$,
 \begin{align*}
  \Phi(n) &\le  \frac{e^{\lambda (n-1)}}{n} \sum_{j=1}^{n} \Phi(\din(x_j))\\
  &\le \frac{e^{\lambda n}}{n} \sum_{j=1}^{n} \Phi(\din(x_j)).
 \end{align*}
 Since $\din(x_j)<n$, using induction,
 \begin{align}
 \label{eqn:wff}
  \frac{e^{\lambda n}}{n} \sum_{j=1}^{n} \Phi(\din(x_j)) \le \frac{e^{\lambda n}}{n} \sum_{j=1}^{n} e^{k'\lambda \din(x_j)}.
 \end{align}
Observe that $e^{k'\lambda \din(x_j)}$ is a convex function of
$\din(x_j)$ and $\sum_{j=1}^n \din(x_j)=\frac{n(n-1)}{2}$. As a
result, the RHS of~\eqref{eqn:wff} is maximized when the in-degrees
take their extreme values, \ie any permutation of
$(0,1,\ldots,n-1)$. Therefore,
\begin{align*}
 \frac{e^{\lambda n}}{n} \sum_{j=1}^{n} e^{k'\lambda \din(x_j)} &\le \frac{e^{\lambda n}}{n} \sum_{j=0}^{n-1} e^{k'\lambda j}\\
 &=\frac{e^{\lambda n}}{n} \frac{e^{k'\lambda n}-1}{e^{k'\lambda}-1}.
 \end{align*}
Combining the above equations,
\begin{align*}
 \Phi(n)\le \frac{e^{\lambda n}}{n} \frac{e^{k'\lambda n}-1}{e^{k'\lambda}-1}.
\end{align*}
Similarly, by induction on $1\le i<n$,
\begin{align*}
 \Phi(i)\le \frac{e^{\lambda i}}{i} \frac{e^{k'\lambda i}-1}{e^{k'\lambda}-1}.
\end{align*}
In Lemma~\ref{haha} we show that for $1 \le i\le n$,
\begin{align}
\label{eqn:ke}
 \frac{e^{\lambda i}}{i} \frac{e^{k'\lambda i}-1}{e^{k'\lambda}-1} \le e^{k'\lambda i}.
\end{align}
Therefore, $\Phi(i)\le e^{k'\lambda i}$ for $1\le i\le n$ and in
particular, $\Phi(n)\le e^{k'\lambda n}$. Substituting $\EE[e^{\lambda
    Q_n}]=\Phi(n)$ in~\eqref{eqn:chernoff},
\begin{align*}
 \pr(Q_n>kn) &\le \frac{e^{k'\lambda n}}{e^{k\lambda n}}\\
 &=\frac{e^{k'\ln k'}}{e^{k\ln k'}}\\
 &=e^{-(k-k')\ln k'}.
\end{align*}
\noindent This proves the Lemma.
\qed

\noindent We now prove~\eqref{eqn:ke}.
Let $k'=\max\{e,\frac{k}{2}\}$ and $\lambda=\frac 1n\ln k'$. 
\begin{lemma}
\label{haha}
 For all $1\le i\le n$, $\frac{e^{\lambda i}}{i} \frac{e^{k'\lambda
     i}-1}{e^{k'\lambda}-1}\le e^{k'\lambda i} $.
\end{lemma}
\begin{proof}
 It suffices to show that for all $0<t\le n$,
\begin{align*}
 f(t)&\df\frac{e^{\lambda t}}{t} \frac{1-e^{-k'\lambda t}}{e^{k'\lambda}-1}<1.
\end{align*}
Observe that,
\begin{align*}
 \lim_{t\to 0} f(t)=\frac{k'\lambda}{e^{k'\lambda}-1} \le 1.
\end{align*}
On the other hand,
\begin{align*}
 f(n)&=\frac{e^{\lambda n}}{n} \frac{1-e^{-k'\lambda n}}{e^{k'\lambda}-1}\\
 & \le \frac{k'}{n}\frac{1}{e^{k'\ln k'/n}-1}\\
 & \le \frac{k'}{n}\frac{n}{k'\ln k'}\\
 & \le 1.
\end{align*}
Next, we show that $f(t)$ is convex. One can show that,
\begin{align*}
 \ln \frac{1-e^{-u}}{u},
\end{align*}
is a convex function of $u$. As a result,
\begin{align*}
 \ln \frac{1-e^{-k'\lambda t}}{t},
\end{align*}
is a convex function of $t$. Observe that $\ln e^{\lambda t}$ is also
convex. Therefore,
\begin{align*}
 \ln \frac{1-e^{-k'\lambda t}}{t}+\ln e^{\lambda t},
\end{align*}
is convex. As a result, logarithm of $f(t)$ is convex and therefore,
$f(t)$ is convex.

We showed that $f(t)$ is convex, $f(t\to 0)\le 1$, and
$f(n)\le1$. Therefore, for all $0<t\le n$, $f(t)\le 1$.
\end{proof}
\end{document}